\documentclass[a4paper,11pt]{article}
\usepackage[]{fontenc}
\usepackage[latin9]{inputenc}
\usepackage{authblk}
\usepackage{array}
\usepackage{amsthm}
\usepackage{amsmath}
\usepackage{amssymb}
\usepackage{graphicx}
\usepackage{esint}
\usepackage[]{geometry}
\newtheorem{thm}{Theorem}[section]

\newtheorem{lemma}{Lemma}[section]

\newtheorem{assumption}{Assumption}
\newtheorem{remark}{Remark}

\begin{document}   

\title{Unbiased estimators for the Heston model with stochastic interest rates}

\date{}

\author{Chao Zheng \thanks{Email: chao.zheng12@gmail.com.}}
\author{Jiangtao Pan \thanks{Email: zufe170112600119@163.com}}
\affil{School of Data Sciences, Zhejiang University of Finance and Economics, Hangzhou, China}

\maketitle

\begin{abstract}
We combine the unbiased estimators in Rhee and Glynn (Operations Research: 63(5), 1026-1043, 2015) and the Heston model with stochastic interest rates. Specifically, we first develop a semi-exact log-Euler scheme for the Heston model with stochastic interest rates. Then, under mild assumptions, we show that the convergence rate in the $L^2$ norm is $O(h)$, where $h$ is the step size. The result applies to a large class of models, such as the Heston-Hull-While model, the Heston-CIR model and the Heston-Black-Karasinski model.  Numerical experiments support our theoretical convergence rate.
\end{abstract}



\section{Introduction}
The classical Heston model (Heston \cite{He}) is one of the fundamental models in finance, and it has been widely applied in various financial markets, such as the equity, fixed income and foreign exchange markets, due to its tractability in modelling the term structure of implied volatility. However, in this model, the interest rate is constant, and this assumption is often not appropriate for long-maturity options, because the long-term behaviour of the interest rate is typically far from constant. This phenomenon has been empirically investigated by Bakshi, Cao and Chen \cite{BCC}. A natural extension of the Heston model is to incorporate a stochastic interest rate, which is usually referred to as the Heston model with stochastic interest rates. There are several contributions in this direction, such as Grzelak and Oosterlee \cite{GO}, Van Haastrecht and Pelsser \cite{VHP} and references therein. 

Under the Heston model with stochastic interest rates, the price of an option can be written as
\[
\mathbb{E}\left(e^{-\int_{0}^{T}r_{t}\mathrm{d}t}P(S)\right)
\]
where $S$ is the solution to the Heston model with stochastic interest rates and $r_{t}$ is the underlying interest rate. Here, $P$ denotes the payoff functional of an option. We are interested in calculating this expectation, as for the majority of options, there are no closed-form formulas. A typical approach is to use a Monte Carlo method associated with a time-discrete scheme on $S$ for an approximate value.  Recently, Rhee and Glynn \cite{RG} proposed several unbiased Monte Carlo estimators based on a randomization idea for a stochastic differential equation, which are unbiased versions of multilevel Monte Carlo estimators by Giles \cite{Gi}. These unbiased estimators have a clear advantage over the standard Monte Carlo estimator, as the latter is typically biased when $S$ has to be approximated through a time-discrete scheme. To combine Rhee and Glynn's estimators and the Heston model with stochastic interest rates, it is essential to develop a numerical scheme with a sufficiently high convergence rate in the $L^2$ norm. However, standard theorems, such as those in Kloeden and Platen \cite{KP}, require that the drift and diffusion coefficients satisfy the global Lipschitz and linear growth assumptions, which are not satisfied by the Heston model with stochastic interest rates. Research on developing time-discrete schemes for the Heston model with stochastic interest rates is scarce, which is challenging. Cozma, Mariapragassam and Reisinger \cite{CMR} proposed a different log-Euler scheme for the stochastic-local volatility model with stochastic rates and they demonstrated strong convergence without providing a rate. This is the only reference we are aware of concerning the convergence of Monte Carlo algorithms for the Heston model with stochastic interest rates. 

In this article, we develop a semi-exact log-Euler scheme for the Heston model with stochastic interest rates, where the driven Brownian motion for interest rate models is independent of the driven Brownian motions for the Heston component. The scheme is an extension of those in Mickel and Neuenkirch \cite{MN} and Zheng \cite{CZ} \cite{CZ1} for the classical Heston model. Under mild assumptions on interest rate models, we show that the underlying scheme converges with order one in the $L^2$ norm. The types of options we consider include those with bounded and Lipschitz continuous payoffs and the digital option with a discontinuous payoff. When a numerical scheme has a convergence rate higher than $1/2$, it is convenient to combine it into Rhee and Glynn's unbiased estimators. Our result applies to a large class of models, including the Heston-Hull-While model, the Heston-CIR model and the Heston-Black-Karasinski model, among which the Heston-Hull-White model and the Heston-CIR model are particularly attractive in practical applications. Furhermore, we extend the result to the Heston model with stochastic interest rates, where the driven Brownian motion for interest rate models is correlated with that for the Heston component 

There are two advantages of the scheme we develop. One is that the convergence rate is free of Feller's index (except for the digital option), i.e., the convergence rate is valid for the full range of parameters in the Heston model with stochastic interest rates. The other is that the convergence rate is higher than the usual convergence rate ($1/2$ in the $L^2$ norm) of the standard Euler scheme under standard assumptions, where the usual rate may not be sufficiently high for the unbiased estimation.

The remainder of this article is organized as follows. In section 2, we review the Heston model with stochastic interest rates and develop a log-Euler scheme for it. Section 3 reviews the unbiased estimators from Rhee and Glynn \cite{RG}. In section 4, we derive the relevant convergence rate under several mild assumptions. Section 5 extends the result to a more general model. Section 6 illustrates numerical results to support our theoretical analysis. Finally, we conclude in section 7.

\section{Heston model with stochastic interest rates}\label{Hestonstr} \label{section2}
Let $(\Omega,\mathcal{F}, (\mathcal{F}_{t})_{t\geq 0}, \mathbb{P})$ be a filtered probability space satisfying the usual assumptions. The Heston model with stochastic interest rates is of the form
\begin{align*}
\mathrm{d}S_{t}&=r_{t}S_{t}\mathrm{d}t+\sqrt{V_{t}}S_{t}(\rho\mathrm{d}W_{t}^{1}+\sqrt{1-\rho^{2}}\mathrm{d}W_{t}^{2})\\
\mathrm{d}V_{t}&=k(\theta-V_{t})\mathrm{d}t+\sigma\sqrt{V_{t}}\mathrm{d}W_{t}^{1},
\end{align*}
where $(W^1_{t})_{t\geq 0}$ and $(W^2_{t})_{t\geq 0}$ are two independent $\mathcal{F}_{t}$-adapted Brownian motions and the parameters $k,\theta,\sigma>0$ and $\rho\in [-1,1]$. The initial values $S_0, V_0>0$. Here, $(r_{t})_{t\geq 0}$ is a stochastic interest rate. If we replace $r_t$ in the model with a constant, then it becomes the classical Heston model. The classical interest rate models and their generalizations can be found in Brigo and Mercurio \cite{BM}. Among them, a large class of interest rate models can be written as
\[
\mathrm{d}r_{t}=\mu(t,r_{t})\mathrm{d}t+\phi(t,r_{t})\mathrm{d}W_{t}^{3},
\]
where $\mu, \phi: [0,T]\times \mathbb{R}\rightarrow \mathbb{R}$ are continuous functions and $(W_{t}^{3})_{t\geq 0}$ is a $\mathcal{F}_{t}$-adapted Brownian motion. We assume that $(W_{t}^{3})_{t\geq 0}$ is independent of $(W_{t}^{1})_{t\geq 0}$ and $(W_{t}^{2})_{t\geq 0}$.  Furthermore, we assume that there is a unique solution to the equation of $(r_{t})_{t\geq 0}$ above. 

Let $X_{t}=\ln(S_{t})$. By using It$\hat{o}$'s lemma, we have
\[
\mathrm{d}X_{t}=\left(r_{t}-\frac{1}{2}V_{t}\right)\mathrm{d}t+\sqrt{V_{t}}\left(\rho\mathrm{d}W_{t}^{1}+\sqrt{1-\rho^{2}}\mathrm{d}W_{t}^{2}\right).
\]
Then substituting the equation of $V_{t}$ into the equation above, we obtain
\[
\mathrm{d}X_{t}=\left(\left(r_{t}-\frac{k\rho\theta}{\sigma}\right)+\left(\frac{k\rho}{\sigma}-\frac{1}{2}\right)V_{t}\right)\mathrm{d}t+\frac{\rho}{\sigma}\mathrm{d}V_{t}+\sqrt{1-\rho^{2}}\sqrt{V_{t}}\mathrm{d}W_{t}^{2}.
\]
Since $(V_{t})_{t\geq 0}$ is independent of $(W^2_{t})_{t\geq 0}$, the stochastic integral $\int_{0}^{T}\sqrt{V_{t}}\mathrm{d}W_{t}^{2}$ is normally distributed with mean $0$ and variance $\int_{0}^{T}V_{t}\mathrm{d}t$. Therefore, the solution at any finite time horizon $T>0$ can be written as 
\begin{align}
X_{T}&=X_{0}+\left[\int_{0}^{T}r_{t}\mathrm{d}t+\left(\frac{\rho k}{\sigma}-\frac{1}{2}\right)\int_{0}^{T}V_{t}\mathrm{d}t+\frac{\rho}{\sigma}\left(V_{T}-V_{0}-k\theta T\right)\right.\nonumber\\
&\quad\left.+\sqrt{1-\rho^{2}}\sqrt{\int_{0}^{T}V_{t}\mathrm{d}t}N\right]\label{XT}
\end{align}
where $N$ is a standard normal random variable, that is independent of $(V_{t})_{t\in [0,
T]}$. Note that $N$ is also independent of $(r_{t})_{t\in [0,T]}$, because the driving Brownian motion $(W^3_{t})_{t\in [0,T]}$ of $(r_{t})_{t\in [0,T]}$ is independent of $(W^2_{t})_{t\in [0,T]}$ and $(V^2_{t})_{t\in [0,T]}$. There are several integrals in equation (\ref{XT}) to be approximated.

It is known that $V_{t}$ follows a scaled noncentral chi-squared distribution given $V_{u}$ for any $u\in [0,t)$, i.e.,
\[
V_{t}\overset{\mathrm{d}}{=}\frac{\sigma^{2}(1-\mathrm{e}^{-k(t-u)})}{4k}\chi_{d}^{2}\left(\frac{4k\mathrm{e}^{-k(t-u)}}{\sigma^{2}(1-\mathrm{e}^{-k(t-u)})}V_{u}\right),
\]
where $\chi_{d}^{2}(\lambda)$ denotes a noncentral chi-squared random variable with degrees of freedom $d=\frac{4k \theta}{\sigma^{2}}>0$  and noncentrality parameter $\lambda>0$ (see Glasserman \cite{G}). Hence, $V_{t}$ can be sampled exactly. Let $(\hat{r}_{ih})_{i=1,..,T/h}$ be an approximate path of $(r_{t})_{t\in [0,T]}$. It is convenient to approximate
\[
\int_{0}^{T}r_{t}\mathrm{d}t\approx \sum_{i=0}^{T/h-1}\hat{r}_{ih}h, \quad \int_{0}^{T}V_{t}\mathrm{d}t\approx \sum_{i=0}^{T/h-1}V_{ih}h,
\]
by using the Euler scheme based on step size $h$. We denote by $\hat{X}_{T}^{h}$ the approximated solution of $X_{T}$. Let $\hat{S}_{T}^h:=e^{\hat{X}_{T}^h}$, and then $\hat{S}_{T}^h$ is an approximation of $S_{T}$.

\section{Unbiased estimators for SDEs} \label{Sect3}
In this section, we review the unbiased estimators introduced in Rhee and Glynn \cite{RG}. The prices of many options can be expressed as 
\[
\mathbb{E}(Y):=\mathbb{E}\left(e^{-\int_{0}^{T}r_{t}\mathrm{d}t}P(S_{T})\right)
\]
where $P$ is the payoff function and $Y\in L^2$ (i.e., $\mathbb{E}(Y^2)<\infty$). To estimate this expectation, Rhee and Glynn \cite{RG} proposed an estimator
\[
Z=\sum_{n=0}^{\mathcal{N}}\frac{\Delta_{n}}{\mathbb{P}(\mathcal{N}\geq n)}
\]
where $\Delta_{n}=Y_{n}-Y_{n-1}$ and $Y_{n}$, $n\in \mathbb{N}$, approximates $Y$ with step size $T/2^n$ and $Y_{-1}=0$. In this article, we let
\[
Y_{n}=e^{-\sum_{i=0}^{T/h-1}\hat{r}_{ih}h}P(\hat{S}_{T}^h), \quad h=T/2^n.
\]
Here, $\mathcal{N}$ is a nonnegative integer-valued random variable that is independent of $Y_{n}$. This estimator is usually referred to as the coupled sum estimator. In addition, Rhee and Glynn \cite{RG} proposed the single term estimator
\[
\tilde{Z}=\frac{\Delta_{\mathcal{N}}}{p_\mathcal{N}},
\]
where $p_n=\mathbb{P}(\mathcal{N}=n)$. 

For the coupled sum estimator, suppose that $Y_{n}$ converges to $Y$ in the $L^2$ norm as $n\rightarrow \infty$. Theorem 1 of Rhee and Glynn \cite{RG} showed that if 
\begin{equation}
\sum_{n=0}^{\infty}\frac{\mathbb{E}\left[(Y_{n-1}-Y)^2\right]}{\mathbb{P}(\mathcal{N}\geq n)}<\infty \label{var}
\end{equation}
then $Z$ is an unbiased estimator of $\mathbb{E}(Y)$ (i.e., $\mathbb{E}(Z)=\mathbb{E}(Y)$) with a finite variance. Furthermore, the average computational time of $Z$ is proportional to 
\begin{equation}
\sum_{n=0}^{\infty}2^n\mathbb{P}(\mathcal{N}\geq n). \label{ct}
\end{equation}
Therefore, if $\mathbb{E}\left[(Y_{n}-Y)^2\right]=O(2^{-2np})$ with $p>1/2$ (here, $p$ is the convergence rate in the $L^2$ norm), we can easily construct a distribution for $\mathcal{N}$ such that $\mathbb{P}(\mathcal{N}\geq n)=O(2^{-n(p+1/2)})$ to ensure that (\ref{var}) and (\ref{ct}) are finite. The optimal distribution of $\mathcal{N}$ can be calculated by minimizing the product of the variance and the average computational time of $Z$; see Rhee and Glynn \cite{RG}, Cui, et.al \cite{Cui} and Zheng, Pan and Wang \cite{ZPW}. 

For the single term estimator, if 
\[
\sum_{n=0}^{\infty}\frac{\mathbb{E}\left[(Y_{n-1}-Y_n)^2\right]}{p_n}<\infty 
\]
then $\tilde{Z}$ is an unbiased estimator with a finite variance. Similarily, if $\mathbb{E}\left[(Y_{n}-Y)^2\right]=O(2^{-2np})$, $p>1/2$, it is convenient to set $p_n=O(2^{-n(p+1/2)})$. The optimal distribution of $\mathcal{N}$ can be found at Rhee and Glynn \cite{RG} and Zheng, Pan and Wang \cite{ZPW}. 

Therefore, it is important to investigate the convergence rate of $\mathbb{E}\left[(Y_{n}-Y)^2\right]$.

\section{Convergence analysis} \label{sect2}
Recall that the interest rate $(r_{t})_{t\in [0,T]}$ follows the stochastic differential equation
\[
\mathrm{d}r_{t}=\mu(t,r_{t})\mathrm{d}t+\phi(t,r_{t})\mathrm{d}W_{t}^{3},
\]
where $(W_{t}^{3})_{t\in [0,T]}$ is independent of $(W_{t}^{1})_{t\in [0,T]}$ and $(W_{t}^{2})_{t\in [0,T]}$. Let $c$ and $c_{q}$ be constants regardless of their values, where $c_{q}$ relies on $q\geq 1$. Our analysis throughout this article is based on the following assumption:
\begin{assumption} \label{as1}
For any $q\geq 1$, it holds that
\[
\sup_{t\in [0,T]}\mathbb{E}\left[|\mu(t,r_{t})|^{q}\right]<\infty, \quad \sup_{t\in [0,T]}\mathbb{E}\left[|\phi(t,r_{t})|^{q}\right]<\infty,
\]
and the approximate interest rate $(\hat{r}_{ih})_{i=1,..,T/h}$ satisfies
\[
\max_{i=1,..,T/h}\mathbb{E}\left[|\hat{r}_{ih}-r_{ih}|^{q}\right]<c_{q}h^{q}.
\]
\end{assumption}
\begin{remark}
A typical time-discrete scheme that may satisfy Assumption \ref{as1} is the Milstein scheme. Under some standard assumptions on model coefficients (e.g., Lipschitz continuity, linear growth), the convergence rate in the $L^2$ norm is one. 
\end{remark}

\begin{lemma}\label{lemma1}
Under Assumption \ref{as1}, we obtain
\[
\mathbb{E}\left[\left|\int_{0}^{T}r_{t}\mathrm{d}t-\sum_{i=0}^{T/h-1} \hat{r}_{ih}h\right|^{q}\right]=O(h^{q}), \quad \forall q\geq 2.
\]
\end{lemma}
\begin{proof}
We see that
\begin{align}
&\mathbb{E}\left[\left|\int_{0}^{T}r_{t}\mathrm{d}t-\sum_{i=0}^{T/h-1} \hat{r}_{ih}h\right|^{q}\right]\nonumber\\
&\leq c_{q}\mathbb{E}\left[\left|\int_{0}^{T}r_{t}\mathrm{d}t-\sum_{i=0}^{T/h-1} r_{ih}h\right|^{q}\right]+c_{q}\mathbb{E}\left[\left|\sum_{i=0}^{T/h-1} (\hat{r}_{ih}-r_{ih})h\right|^{q}\right].\label{eq2.1}
\end{align}
Let $\eta(t):=\max\{lh:lh\leq t, l=0,1,2,...\}$. For the first term of (\ref{eq2.1}), we have
\begin{align}
&\mathbb{E}\left[\left|\int_{0}^{T}r_{t}\mathrm{d}t-\sum_{i=0}^{T/h-1} r_{ih}h\right|^{q}\right]\nonumber\\
&=\mathbb{E}\left[\left|\int_{0}^{T}r_{\eta(t)}\mathrm{d}t-\int_{0}^{T}r_{t}\mathrm{d}t\right|^{q}\right]\nonumber\\
&\leq c_{q}\mathbb{E}\left[\left|\int_{0}^{T}\left(\int_{\eta(t)}^{t}\mu(u,r_{u})\mathrm{d}u\right)\mathrm{d}t\right|^{q}\right]+c_{q}\mathbb{E}\left[\left|\int_{0}^{T}\left(\int_{\eta(t)}^{t}\phi(u,r_{u})\mathrm{d}W_{u}^{3}\right)\mathrm{d}t\right|^{q}\right].\label{Eulerind1}
\end{align}
An application of the Fubini theorem yields 
\begin{align*}
\int_{0}^{T}\left(\int_{\eta(t)}^{t}\mu(u,r_{u})\mathrm{d}u\right)\mathrm{d}t&=\int_{0}^{T}\left(\int_{u}^{\eta(u)+h}\mu(u,r_{u})\mathrm{d}t\right)\mathrm{d}u\\
&=h\int_{0}^{T}\left(1+\frac{\eta(u)-u}{h}\right)\mu(u,r_{u})\mathrm{d}u.
\end{align*}
Hence, it follows that
\begin{equation}
\mathbb{E}\left[\left|\int_{0}^{T}\left(\int_{\eta(t)}^{t}\mu(u,r_{u})\mathrm{d}u\right)\mathrm{d}t\right|^{q}\right]\leq c_{q}h^{q}\int_{0}^{T}\mathbb{E}[|\mu(u,r_{u})|^{q}]\mathrm{d}u=O(h^{q}).\label{eq2.2}
\end{equation}
Furthermore, we obtain from the stochastic Fubini theorem (see Theorem 65, Protter \cite{Pr}), the Burkholder-Davies-Gundy inequality and the Cauchy-Schwarz inequality that
\begin{align}
&\mathbb{E}\left[\left|\int_{0}^{T}\left(\int_{\eta(t)}^{t}\phi(u,r_{u})\mathrm{d}W_{u}^{3}\right)\mathrm{d}t\right|^{q}\right]\nonumber\\
&= h^{q}\mathbb{E}\left[\left|\int_{0}^{T}\left(1+\frac{\eta(u)-u}{h}\right)\phi(u,r_{u})\mathrm{d}W_{u}^{3}\right|^{q}\right]\nonumber\\
&\leq c_{q}h^{q}\mathbb{E}\left[\left|\int_{0}^{T}\phi^2(u,r_{u})\mathrm{d}u\right|^{q/2}\right]\leq c_{q}h^{q}\int_{0}^{T}\mathbb{E}[|\phi(u,r_{u})|^{q}]\mathrm{d}u=O(h^{q}).\label{eq2.3}
\end{align}
Therefore, substituting (\ref{eq2.2}) and (\ref{eq2.3}) into (\ref{Eulerind1}), we show that (\ref{Eulerind1}) is $O(h^{q})$. For the second term of (\ref{eq2.1}), by Jensen's inequality, we have
\begin{equation}
\mathbb{E}\left[\left|\sum_{i=0}^{T/h-1} (\hat{r}_{ih}-r_{ih})h\right|^{q}\right]\leq c_q\sum_{i=0}^{T/h-1}\mathbb{E}\left[|\hat{r}_{ih}-r_{ih}|^{q}\right]h=O(h^{q}).\label{eq2.4}
\end{equation}
Thus, combining (\ref{Eulerind1}) and (\ref{eq2.4}) into (\ref{eq2.1}), the proof is complete. 
\end{proof}

\begin{lemma}\label{lemma2}
For the stochastic process $(V_{t})_{t\in [0,T]}$, we have
\[
\mathbb{E}\left[\left|\sqrt{\int_{0}^{T}V_{t}\mathrm{d}t}-\sqrt{\sum_{i=0}^{T/h-1} V_{ih}h}\right|^{q}\right]=O(h^{q}), \quad \forall q\geq 1.
\]
\end{lemma}
\begin{proof}
The Cauchy-Schwarz inequality implies that
\begin{align}
&\mathbb{E}\left[\left|\sqrt{\int_{0}^{T}V_{t}\mathrm{d}t}-\sqrt{\sum_{i=0}^{T/h-1} V_{ih}h}\right|^{q}\right]\nonumber\\
&=\mathbb{E}\left[\left|\frac{\int_{0}^{T}V_{t}\mathrm{d}t-\sum_{i=0}^{T/h-1} V_{ih}h}{\sqrt{\int_{0}^{T}V_{t}\mathrm{d}t}+\sqrt{\sum_{i=0}^{T/h-1} V_{ih}h}}\right|^{q}\right]\nonumber\\
&\leq \mathbb{E}\left[\left|\frac{\int_{0}^{T}V_{t}\mathrm{d}t-\sum_{i=0}^{T/h-1} V_{ih}h}{\sqrt{\int_{0}^{T}V_{t}\mathrm{d}t}}\right|^{q}\right]\nonumber\\
&\leq \sqrt{\mathbb{E}\left[\left|\int_{0}^{T}V_{t}\mathrm{d}t-\sum_{i=0}^{T/h-1} V_{ih}h\right|^{2q}\right]}\cdot \sqrt{\mathbb{E}\left[\left(\frac{1}{\int_{0}^{T}V_{t}\mathrm{d}t}\right)^{q}\right]}.\label{xt2}
\end{align}
From Theorem 4.1(a) in Dufresne \cite{Du}, we learn that 
\[
\mathbb{E}\left[\left(\frac{1}{\int_{0}^{T}V_{t}\mathrm{d}t}\right)^{q}\right]<\infty 
\]
for any $q\in \mathbb{R}$. Since the coefficients in the equation of $(V_{t})_{t\in [0,T]}$ satisfy Assumption \ref{as1}, we conclude from Lemma \ref{lemma1} that $\mathbb{E}\left[\left|\int_{0}^{T}V_{t}\mathrm{d}t-\sum_{i=0}^{T/h-1}V_{ih}h\right|^{2q}\right]=O(h^{2q})$. Therefore, the term of (\ref{xt2}) is $O(h^{q})$ and the proof is complete.
\end{proof}

\begin{thm} \label{thm2.1}
Under Assumption \ref{as1}, we have
\[
\mathbb{E}\left[\left|X_{T}-\hat{X}_{T}^{h}\right|^{q}\right]=O(h^{q}), \quad \forall q\geq 2.
\]
\end{thm}
\begin{proof}
Straightforward calculation shows that
\begin{align*}
&\mathbb{E}\left[\left|X_{T}-\hat{X}_{T}^{h}\right|^{q}\right]\\
&\leq c_{q}\mathbb{E}\left[\left|\int_{0}^{T}r_{t}\mathrm{d}t-\sum_{i=0}^{T/h-1} \hat{r}_{ih}h\right|^{q}\right]+c_{q}\mathbb{E}\left[\left|\int_{0}^{T}V_{t}\mathrm{d}t-\sum_{i=0}^{T/h-1} V_{ih}h\right|^{q}\right]\\
&\quad +c_{q}\mathbb{E}\left[\left|\sqrt{\int_{0}^{T}V_{t}\mathrm{d}t}-\sqrt{\sum_{i=0}^{T/h-1} V_{ih}h}\right|^{q}\right].
\end{align*}
and an application of Lemmas \ref{lemma1} and \ref{lemma2} completes the proof.
\end{proof}

Then we proceed to the convergence rate of the log-Euler scheme to approximate the price of an option. Different types of options may have different payoff functions, which can be continuous or discontinuous. 

\subsection{Analysis for continuous payoffs}
For continuous payoffs, we impose an assumption: 

\begin{assumption} \label{as2}
The payoff $P:[0,+\infty)\rightarrow \mathbb{R}$ is Lipschitz continuous and there exists a constant $C>0$, such that $P(U)=P(C)$ for all $U>C$.
\end{assumption}

Under Assumption \ref{as2}, it holds that
\begin{equation}
|P(U_{1})-P(U_{2})|\leq c|\ln U_{1} - \ln U_{2}| \label{eq2.5}
\end{equation}
for all $U_{1},U_{2}\in [0,+\infty)$, see Theorem 3.1 in Zheng \cite{CZ}. The payoff that satisfies Assumption \ref{as2} is bounded, which is typically suitable for a put-style option. The put-style option becomes worthless when the price of the underlying asset $S_{T}$ is sufficiently high.  For example, the standard European put option has payoff $P(S_{T}):=\max\{K-S_{T},0\}$ with $K>0$, which satisfies Assumption \ref{as2}.

\begin{thm}\label{thm1}
Suppose that Assumptions \ref{as1} and \ref{as2} are satisfied. Suppose that there exists $p>1$, such that $\mathbb{E}\left(e^{-2p\int_{0}^{T}r_{t}\mathrm{d}t}\right)< \infty$ and $\sup_h\mathbb{E}\left(e^{-2p\sum_{i=0}^{T/h-1}\hat{r}_{ih}h}\right)< \infty$. Then, we have
\[
\mathbb{E}\left[\left(e^{-\int_{0}^{T}r_{t}\mathrm{d}t}P(S_{T})\right)^2\right]<\infty
\]
and
\[
\mathbb{E}\left[\left(e^{-\int_{0}^{T}r_{t}\mathrm{d}t}P(S_{T})-e^{-\sum_{i=0}^{T/h-1}\hat{r}_{ih}h}P(\hat{S}_{T}^h)\right)^2\right]=O(h^2).
\]
\end{thm}
\begin{proof}
For the first term, it follows from Jensen's inequality and the boundedness of $P(S_T)$ that
\[
\mathbb{E}\left[\left(e^{-\int_{0}^{T}r_{t}\mathrm{d}t}P(S_{T})\right)^2\right]<c\mathbb{E}\left(e^{-2\int_{0}^{T}r_{t}\mathrm{d}t}\right)<c\left[\mathbb{E}\left(e^{-2p\int_{0}^{T}r_{t}\mathrm{d}t}\right)\right]^{1/p}<\infty,
\]
where $p>1$. Then, we focus on the second term. The Taylor expansion, together with H$\ddot{o}$lder's inequality, gives
\begin{align*}
&\mathbb{E}\left[\left(e^{-\int_{0}^{T}r_{t}\mathrm{d}t}-e^{-\sum_{i=0}^{T/h-1}\hat{r}_{ih}h}\right)^2\right]\\
&=\mathbb{E}\left[e^{-2\varepsilon}\left(\int_{0}^{T}r_{t}\mathrm{d}t-\sum_{i=0}^{T/h-1}\hat{r}_{ih}h\right)^2\right]\\
&\leq \left[\mathbb{E}\left(\max(e^{-2p\int_{0}^{T}r_{t}\mathrm{d}t},e^{-2p\sum_{i=0}^{T/h-1}\hat{r}_{ih}h})\right)\right]^{1/p}\cdot \left[\mathbb{E}\left|\int_{0}^{T}r_{t}\mathrm{d}t-\sum_{i=0}^{T/h-1}\hat{r}_{ih}h\right|^{2q}\right]^{1/q}
\end{align*}
where $\frac{1}{p}+\frac{1}{q}=1$, and $\varepsilon$ is between $\int_{0}^{T}r_{t}\mathrm{d}t$ and $\sum_{i=0}^{T/h-1}\hat{r}_{ih}h$. Lemma \ref{lemma1} shows that
\[
\mathbb{E}\left[\left|\int_{0}^{T}r_{t}\mathrm{d}t-\sum_{i=0}^{T/h-1}\hat{r}_{ih}h\right|^{2q}\right]=O(h^{2q}).
\]
Note that
\[
\mathbb{E}\left(\max(e^{-2p\int_{0}^{T}r_{t}\mathrm{d}t},e^{-2p\sum_{i=0}^{T/h-1}\hat{r}_{ih}h})\right)\leq \mathbb{E}\left(e^{-2p\int_{0}^{T}r_{t}\mathrm{d}t}\right)+\sup_h\mathbb{E}\left(e^{-2p\sum_{i=0}^{T/h-1}\hat{r}_{ih}h}\right)<c,
\]
where $c$ is a constant independent of $h$. Consequently
\[
\mathbb{E}\left[\left(e^{-\int_{0}^{T}r_{t}\mathrm{d}t}-e^{-\sum_{i=0}^{T/h-1}\hat{r}_{ih}h}\right)^2\right]=O(h^2).
\]
Therefore, using the boundedness of $P$, H$\ddot{o}$lder's inequality, inequality (\ref{eq2.5}) and Theorem \ref{thm2.1}, we have
\begin{align*}
&\mathbb{E}\left[\left(e^{-\int_{0}^{T}r_{t}\mathrm{d}t}P(S_{T})-e^{-\sum_{i=0}^{T/h-1}\hat{r}_{ih}h}P(\hat{S}_{T}^h)\right)^2\right]\\
&=\mathbb{E}\left[\left(e^{-\int_{0}^{T}r_{t}\mathrm{d}t}\left(P(S_{T})-P(\hat{S}_{T}^h)\right)+P(\hat{S}_{T}^h)\left(e^{-\int_{0}^{T}r_{t}\mathrm{d}t}-e^{-\sum_{i=0}^{T/h-1}\hat{r}_{ih}h}\right)\right)^2\right]\\
&\leq 2\mathbb{E}\left[e^{-2\int_{0}^{T}r_{t}\mathrm{d}t}\left(P(S_{T})-P(\hat{S}_{T}^h)\right)^2\right]+2\mathbb{E}\left[P^2(\hat{S}_{T}^h)\left(e^{-\int_{0}^{T}r_{t}\mathrm{d}t}-e^{-\sum_{i=0}^{T/h-1}\hat{r}_{ih}h}\right)^2\right]\\
&\leq c\left[\mathbb{E}\left(e^{-2p\int_{0}^{T}r_{t}\mathrm{d}t}\right)\right]^{1/p}\cdot \left[\mathbb{E}\left|P(S_{T})-P(\hat{S}_{T}^h)\right|^{2q}\right]^{1/q}+c\mathbb{E}\left[\left(e^{-\int_{0}^{T}r_{t}\mathrm{d}t}-e^{-\sum_{i=0}^{T/h-1}\hat{r}_{ih}h}\right)^2\right]\\
&\leq c\left[\mathbb{E}\left|\ln(S_{T})-\ln(\hat{S}_{T}^h)\right|^{2q}\right]^{1/q}+c\mathbb{E}\left[\left(e^{-\int_{0}^{T}r_{t}\mathrm{d}t}-e^{-\sum_{i=0}^{T/h-1}\hat{r}_{ih}h}\right)^2\right]=O(h^2),
\end{align*}
which completes the proof.

\end{proof}

Theorem \ref{thm1} is based on Assumption \ref{as2} for bounded and Lipschitz continuous payoffs. Fortunately, for unbounded payoffs, such as the European call option with the payoff $P(S_{T}):=\max\{S_{T}-K,0\}$, $K>0$, our numerical experiment in the next section suggests that a similar conclusion to Theorem \ref{thm1} might still hold, but the rigorous convergence analysis is difficult. 

\subsection{Analysis for digital options}
The digital option may be one of the most popular options in finance that has a discontinuous payoff. Next, we extend the analysis to the digital option. The digital call option has the payoff 
\[
P(S_{T}):=1_{S_{T}>K}
\]
where $K>0$. Here, if $S_{T}>K$, then $P(S_T)=1$; otherwise, $P(S_T)=0$. Due to the discontinuity, a direct application of our time-discrete scheme in Section \ref{Hestonstr} might lead to a slow convergence rate. Therefore, we consider an approach using conditional expectations similar to that in Giles \cite{Gi2} to achieve a higher convergence rate. Specifically, conditional on $V:=(V_{t})_{t\in [0,T]}$ and $r:=(r_{t})_{t\in [0,T]}$, from equation (\ref{XT}), we have
\begin{align}
&\mathbb{E}\left(e^{-\int_{0}^{T}r_{t}\mathrm{d}t}1_{S_{T}>K}\right)\nonumber\\
&=\mathbb{E}\left[\mathbb{E}\left(e^{-\int_{0}^{T}r_{t}\mathrm{d}t}1_{S_{T}>K}|V,r\right)\right]\nonumber\\
&=\mathbb{E}\left[e^{-\int_{0}^{T}r_{t}\mathrm{d}t}\mathbb{P}(S_{T}>K|V,r)\right]\nonumber\\
&=\mathbb{E}\left[e^{-\int_{0}^{T}r_{t}\mathrm{d}t}\Phi\left(\frac{\ln S_{0}-\ln K+\int_{0}^{T}r_{t}\mathrm{d}t+\left(\frac{\rho k}{\sigma}-\frac{1}{2}\right)\int_{0}^{T}V_{t}\mathrm{d}t+\frac{\rho}{\sigma}\left(V_{T}-V_{0}-k\theta T\right)}{\sqrt{1-\rho^{2}}\sqrt{\int_{0}^{T}V_{t}\mathrm{d}t}}\right)\right]\nonumber\\
&=:\mathbb{E}\left[g\left(V_{T},\int_{0}^{T}V_{t}\mathrm{d}t,\int_{0}^{T}r_{t}\mathrm{d}t\right)\right],\label{digital}
\end{align}
where $\Phi$ is the cumulative distribution function of the standard normal distribution. Then, we convert the approximation of $\mathbb{E}\left(e^{-\int_{0}^{T}r_{t}\mathrm{d}t}1_{S_{T}>K}\right)$ to the approximation of $\mathbb{E}\left[g\left(V_{T},\int_{0}^{T}V_{t}\mathrm{d}t,\int_{0}^{T}r_{t}\mathrm{d}t\right)\right]$, where the integrals $\int_{0}^{T}V_{t}\mathrm{d}t$ and $\int_{0}^{T}r_{t}\mathrm{d}t$ can be approximated by $\sum_{i=0}^{T/h-1}V_{ih}h$ and $\sum_{i=0}^{T/h-1}\hat{r}_{ih}h$, respectively, based on the Euler scheme.  

\begin{lemma}\label{lemma4}
Let $h=T/n$, $n\in \mathbb{N}^+$. For any $q\geq-\frac{2k\theta}{\sigma^2}$, we have
\[
\sup_{h}\mathbb{E}\left[\left(\sum_{i=0}^{T/h-1}V_{ih}h\right)^q\right]<\infty.
\]
\end{lemma} 
\begin{proof}
It is known that $\sup_{t\in[0,T]}\mathbb{E}(V_{t}^q)<\infty$ for any $q\geq -\frac{2k\theta}{\sigma^2}$ (see Theorem 3.1, Hurd and Kuznetsov \cite{HK} or Dereich, Neuenkirch and Szpruch \cite{DNS}). If $q\in [0,1]$, an application of Jensen's inequality yields
\[
\sup_{h}\mathbb{E}\left[\left(\sum_{i=0}^{T/h-1}V_{ih}h\right)^q\right]\leq \sup_{h}\left(\sum_{i=0}^{T/h-1}\mathbb{E}(V_{ih})h\right)^q\leq\left(T\sup_{t\in[0,T]}\mathbb{E}(V_{t})\right)^q<\infty.
\]
If $q\in [-\frac{2k\theta}{\sigma^2},0)\cup (1,+\infty)$, using Jensen's inequality again, we have
\[
\left(\frac{1}{T}\sum_{i=0}^{T/h-1}V_{ih}h\right)^q\leq \frac{1}{T} \sum_{i=0}^{T/h-1}V_{ih}^{q}h
\]
It follows that
\[
\sup_{h}\mathbb{E}\left[\left(\sum_{i=0}^{T/h-1}V_{ih}h\right)^q\right]\leq c_q\sup_{t\in [0,T]}\mathbb{E}(V_{t}^q)
<\infty.
\]
The proof is complete.
\end{proof}

\begin{remark}
Dufresne \cite{Du} proved that $\mathbb{E}\left[\left(\int_{0}^{T}V_{t}\mathrm{d}t\right)^{q}\right]<\infty $ for any $q\in\mathbb{R}$. Lemma \ref{lemma4} is consistent with Dufresne's result, since $\sum_{i=0}^{T/h-1}V_{ih}h$ converges almost surely to $\int_{0}^{T}V_{t}\mathrm{d}t$. However, it seems difficult to prove $\sup_{h}\mathbb{E}\left[\left(\sum_{i=0}^{T/h-1}V_{ih}h\right)^{q}\right]<\infty$ when $q<-\frac{2k\theta}{\sigma^2}$.

\end{remark}

\begin{thm}\label{thm2}
Let $g$ be the function defined at (\ref{digital}) corresponding to the payoff $P(S_{T})=1_{S_{T}>K}$. Suppose that Assumption \ref{as1} is satisfied and that there exists $p>1$, such that $\mathbb{E}\left(e^{-2p\int_{0}^{T}r_{t}\mathrm{d}t}\right)<\infty$ and $\sup_h\mathbb{E}\left(e^{-2p\sum_{i=0}^{T/h-1}\hat{r}_{ih}h}\right)<\infty$. Let $\frac{2k\theta}{\sigma^2}>1$. Then we have
\[
\mathbb{E}\left[\left(g\left(V_{T},\int_{0}^{T}V_{t}\mathrm{d}t,\int_{0}^{T}r_{t}\mathrm{d}t\right)-g\left(V_{T},\sum_{i=0}^{T/h-1}V_{ih}h,\sum_{i=0}^{T/h-1}\hat{r}_{ih}h\right)\right)^2\right]=O(h^2).
\]
\end{thm}
\begin{proof}
For notational convenience, we denote by
\[
\phi:=\Phi\left(\frac{\ln S_{0}-\ln K+\int_{0}^{T}r_{t}\mathrm{d}t+\left(\frac{\rho k}{\sigma}-\frac{1}{2}\right)\int_{0}^{T}V_{t}\mathrm{d}t+\frac{\rho}{\sigma}\left(V_{T}-V_{0}-k\theta T\right)}{\sqrt{1-\rho^{2}}\sqrt{\int_{0}^{T}V_{t}\mathrm{d}t}}\right)
\]
and its approximation
\[
\hat{\phi}:=\Phi\left(\frac{\ln S_{0}-\ln K+\sum_{i=0}^{T/h-1}\hat{r}_{ih}h+\left(\frac{\rho k}{\sigma}-\frac{1}{2}\right)\sum_{i=0}^{T/h-1}V_{ih}h+\frac{\rho}{\sigma}\left(V_{T}-V_{0}-k\theta T\right)}{\sqrt{1-\rho^{2}}\sqrt{\sum_{i=0}^{T/h-1}V_{ih}h}}\right),
\]
where $\Phi$ is the cumulative distribution function of the standard normal distribution. Since $|\Phi(x)|<1$ for all $x$, we obtain
\begin{align}
&\mathbb{E}\left[\left(g\left(V_{T},\int_{0}^{T}V_{t}\mathrm{d}t,\int_{0}^{T}r_{t}\mathrm{d}t\right)-g\left(V_{T},\sum_{i=0}^{T/h-1}V_{ih}h,\sum_{i=0}^{T/h-1}\hat{r}_{ih}h\right)\right)^2\right]\nonumber\\
&\leq 2\mathbb{E}\left[\left(e^{-\sum_{i=0}^{T/h-1}\hat{r}_{ih}h}-e^{-\int_{0}^{T}r_{t}\mathrm{d}t}\right)^2 \hat{\phi}^2\right]+2\mathbb{E}\left[e^{-2\int_{0}^{T}r_{t}\mathrm{d}t}(\hat{\phi}-\phi)^2\right]\nonumber\\
&\leq 2\mathbb{E}\left[\left(e^{-\sum_{i=0}^{T/h-1}\hat{r}_{ih}h}-e^{-\int_{0}^{T}r_{t}\mathrm{d}t}\right)^2 \right]+2\mathbb{E}\left[e^{-2\int_{0}^{T}r_{t}\mathrm{d}t}(\hat{\phi}-\phi)^2\right].\label{eqthm21}
\end{align}
For the first term, Theorem \ref{thm1} implies that
\begin{equation}
\mathbb{E}\left[\left(e^{-\sum_{i=0}^{T/h-1}\hat{r}_{ih}h}-e^{-\int_{0}^{T}r_{t}\mathrm{d}t}\right)^2 \right]=O(h^2).\label{eq21}
\end{equation}
For the second term, as the normal cumulative distribution function $\Phi$ is Lipschitz continuous, we have
\begin{align}
&\mathbb{E}\left[e^{-2\int_{0}^{T}r_{t}\mathrm{d}t}(\hat{\phi}-\phi)^2\right]\nonumber\\
&\leq c\mathbb{E}\left[e^{-2\int_{0}^{T}r_{t}\mathrm{d}t}\left(\frac{1}{\sqrt{\sum_{i=0}^{T/h-1}V_{ih}h}}-\frac{1}{\sqrt{\int_{0}^{T}V_{t}\mathrm{d}t}}\right)^2\right]\nonumber\\
&\quad+c\mathbb{E}\left[e^{-2\int_{0}^{T}r_{t}\mathrm{d}t}\left(\frac{V_{T}}{\sqrt{\sum_{i=0}^{T/h-1}V_{ih}h}}-\frac{V_{T}}{\sqrt{\int_{0}^{T}V_{t}\mathrm{d}t}}\right)^2\right]\nonumber\\
&\quad +c\mathbb{E}\left[e^{-2\int_{0}^{T}r_{t}\mathrm{d}t}\left(\sqrt{\sum_{i=0}^{T/h-1}V_{ih}h}-\sqrt{\int_{0}^{T}V_{t}\mathrm{d}t}\right)^2\right]\nonumber\\
&\quad+c\mathbb{E}\left[e^{-2\int_{0}^{T}r_{t}\mathrm{d}t}\left(\frac{\sum_{i=0}^{T/h-1}\hat{r}_{ih}h}{\sqrt{\sum_{i=0}^{T/h-1}V_{ih}h}}-\frac{\int_{0}^{T}r_{t}\mathrm{d}t}{\sqrt{\int_{0}^{T}V_{t}\mathrm{d}t}}\right)^2\right].\label{lastterm}
\end{align}
Let's focus on the last term of (\ref{lastterm}). It follows that
\begin{align}
&\left|\frac{\sum_{i=0}^{T/h-1}\hat{r}_{ih}h}{\sqrt{\sum_{i=0}^{T/h-1}V_{ih}h}}-\frac{\int_{0}^{T}r_{t}\mathrm{d}t}{\sqrt{\int_{0}^{T}V_{t}\mathrm{d}t}}\right|\nonumber\\
&\leq \left|\sum_{i=0}^{T/h-1}\hat{r}_{ih}h\right|\left|\frac{1}{\sqrt{\sum_{i=0}^{T/h-1}V_{ih}h}}-\frac{1}{\sqrt{\int_{0}^{T}V_{t}\mathrm{d}t}}\right|+\frac{1}{\sqrt{\int_{0}^{T}V_{t}\mathrm{d}t}}\left|\sum_{i=0}^{T/h-1}\hat{r}_{ih}h-\int_{0}^{T}r_{t}\mathrm{d}t\right|\nonumber\\
&=\left|\sum_{i=0}^{T/h-1}\hat{r}_{ih}h\right|\left|\frac{\sqrt{\int_{0}^{T}V_{t}\mathrm{d}t}-\sqrt{\sum_{i=0}^{T/h-1}V_{ih}h}}{\sqrt{\sum_{i=0}^{T/h-1}V_{ih}h}\sqrt{\int_{0}^{T}V_{t}\mathrm{d}t}}\right|+\frac{1}{\sqrt{\int_{0}^{T}V_{t}\mathrm{d}t}}\left|\sum_{i=0}^{T/h-1}\hat{r}_{ih}h-\int_{0}^{T}r_{t}\mathrm{d}t\right|,\nonumber
\end{align}
Using the above inequality and  the fact that the stochastic processes $r$ and $\hat{r}$ are independent of $V$, we obtain
\begin{align*}
&\mathbb{E}\left[e^{-2\int_{0}^{T}r_{t}\mathrm{d}t}\left(\frac{\sum_{i=0}^{T/h-1}\hat{r}_{ih}h}{\sqrt{\sum_{i=0}^{T/h-1}V_{ih}h}}-\frac{\int_{0}^{T}r_{t}\mathrm{d}t}{\sqrt{\int_{0}^{T}V_{t}\mathrm{d}t}}\right)^2\right]\\
&\leq c\mathbb{E}\left[e^{-2\int_{0}^{T}r_{t}\mathrm{d}t}\left(\sum_{i=0}^{T/h-1}\hat{r}_{ih}h\right)^2\right]\mathbb{E}\left[\left(\frac{\sqrt{\int_{0}^{T}V_{t}\mathrm{d}t}-\sqrt{\sum_{i=0}^{T/h-1}V_{ih}h}}{\sqrt{\sum_{i=0}^{T/h-1}V_{ih}h}\sqrt{\int_{0}^{T}V_{t}\mathrm{d}t}}\right)^2\right]\\
&\quad + c\mathbb{E}\left[e^{-2\int_{0}^{T}r_{t}\mathrm{d}t}\left(\sum_{i=0}^{T/h-1}\hat{r}_{ih}h-\int_{0}^{T}r_{t}\mathrm{d}t\right)^2\right]\mathbb{E}\left[\left(\frac{1}{\int_{0}^{T}V_{t}\mathrm{d}t}\right)\right].
\end{align*}
Recall that 
$\mathbb{E}\left[\left(\frac{1}{\int_{0}^{T}V_{t}\mathrm{d}t}\right)^{q}\right]<\infty$
for any $q\in\mathbb{R}$. Furthermore, by Jensen's inequality, it holds that
\[
\mathbb{E}\left[\left|\sum_{i=0}^{T/h-1}\hat{r}_{ih}h\right|^{q}\right]\leq c_q \mathbb{E}\left(\sum_{i=0}^{T/h-1}|\hat{r}_{ih}|^q h\right)\leq c_q \max_{i=1,..,T/h}\mathbb{E}\left(|\hat{r}_{ih}|^{q}\right)<c_q
\]
for any $q\geq 1$, where $c_q$ is independent of $h$. Therefore, applying H$\ddot{o}$lder's inequality, together with Lemma \ref{lemma4}, we conclude that 
\begin{align*}
&\mathbb{E}\left[e^{-2\int_{0}^{T}r_{t}\mathrm{d}t}\left(\frac{\sum_{i=0}^{T/h-1}\hat{r}_{ih}h}{\sqrt{\sum_{i=0}^{T/h-1}V_{ih}h}}-\frac{\int_{0}^{T}r_{t}\mathrm{d}t}{\sqrt{\int_{0}^{T}V_{t}\mathrm{d}t}}\right)^2\right]\\
&\leq c\left[\mathbb{E}\left|\frac{\sqrt{\int_{0}^{T}V_{t}\mathrm{d}t}-\sqrt{\sum_{i=0}^{T/h-1}V_{ih}h}}{\sqrt{\int_{0}^{T}V_{t}\mathrm{d}t}}\right|^{2q_1}\right]^{1/q_1}\left[\mathbb{E}\left(\frac{1}{\sum_{i=0}^{T/h-1}V_{ih}h}\right)^{p_{1}}\right]^{1/p_1}\\
&\quad \times\left[\mathbb{E}(e^{-2p\int_{0}^{T}r_{t}\mathrm{d}t})\right]^{1/p}\left[\mathbb{E}\left|\sum_{i=0}^{T/h-1}\hat{r}_{ih}h\right|^{2q}\right]^{1/q}\\
&\quad+c\left[\mathbb{E}(e^{-2p\int_{0}^{T}r_{t}\mathrm{d}t})\right]^{1/p}\left[\mathbb{E}\left|\sqrt{\int_{0}^{T}r_{t}\mathrm{d}t}-\sqrt{\sum_{i=0}^{T/h-1}\hat{r}_{ih}h}\right|^{2q}\right]^{1/q}\mathbb{E}\left(\frac{1}{\int_{0}^{T}V_{t}\mathrm{d}t}\right)\\
&\leq c\left[\mathbb{E}\left|\sqrt{\int_{0}^{T}V_{t}\mathrm{d}t}-\sqrt{\sum_{i=0}^{T/h-1}V_{ih}h}\right|^{4q_1}\mathbb{E}\left(\frac{1}{\int_{0}^{T}V_{t}\mathrm{d}t}\right)^{2q_1}\right]^{1/(2q_1)}\left[\mathbb{E}\left(\frac{1}{\sum_{i=0}^{T/h-1}V_{ih}h}\right)^{p_{1}}\right]^{1/p_1}\\
&\quad \times\left[\mathbb{E}(e^{-2p\int_{0}^{T}r_{t}\mathrm{d}t})\right]^{1/p}\left[\mathbb{E}\left|\sum_{i=0}^{T/h-1}\hat{r}_{ih}h\right|^{2q}\right]^{1/q}\\
&\quad+c\left[\mathbb{E}(e^{-2p\int_{0}^{T}r_{t}\mathrm{d}t})\right]^{1/p}\left[\mathbb{E}\left|\sqrt{\int_{0}^{T}r_{t}\mathrm{d}t}-\sqrt{\sum_{i=0}^{T/h-1}\hat{r}_{ih}h}\right|^{2q}\right]^{1/q}\mathbb{E}\left(\frac{1}{\int_{0}^{T}V_{t}\mathrm{d}t}\right)=O(h^2),
\end{align*}
where $p,q,p_1,q_1$ satisfy $\frac{1}{p}+\frac{1}{q}=1$ and $\frac{1}{p_1}+\frac{1}{q_1}=1$. Note that when $\frac{2k\theta}{\sigma^2}> 1$, $\sup_h\mathbb{E}\left[\left(\frac{1}{\sum_{i=0}^{T/h-1}V_{ih}h}\right)^{p_{1}}\right]<\infty$ (i.e., $p_1<\frac{2k\theta}{\sigma^2}$, see Lemma \ref{lemma4}) can be achieved by choosing $p_1$ sufficiently close to $1$. The other terms of (\ref{lastterm}) can be shown to be $O(h^{2})$ analogously, hence it holds that $\mathbb{E}\left[e^{-2\int_{0}^{T}r_{t}\mathrm{d}t}(\hat{\phi}-\phi)^2\right]=O(h^2)$. Combining this estimate and (\ref{eq21}) into (\ref{eqthm21}), we complete the proof. 
\end{proof}

The analysis above is for the digital call option. For the digital put option with payoff $P(S_{T}):=1_{K>S_T}$, analogous convergence results can be easily derived following analogous techniques and proofs. 

\section{Extension}
In this section, we extend our result to the Heston model with stochastic interest rates, where the driving Brownian motion for the interest rate process $(r_t)_{t\geq 0}$ and that for the asset process $(S_t)_{t\geq 0}$ are correlated. The model can be written as
\begin{align*}
\mathrm{d}S_{t}&=r_{t}S_{t}\mathrm{d}t+\sqrt{V_{t}}S_{t}(\rho_1\mathrm{d}W_{t}^{1}+\rho_2\mathrm{d}W_{t}^{3}+\sqrt{1-\rho_1^2-\rho_2^2}\mathrm{d}W_{t}^{2})\\
\mathrm{d}V_{t}&=k(\theta-V_{t})\mathrm{d}t+\sigma\sqrt{V_{t}}\mathrm{d}W_{t}^{1}\\
\mathrm{d}r_{t}&=\mu(t,r_{t})\mathrm{d}t+\phi(t,r_{t})\mathrm{d}W_{t}^{3},
\end{align*}
where $\rho_1, \rho_2\in [-1,1]$, and $(W^i_{t})_{t\geq 0}$, $i=1,2,3$, are mutually independent Brownian motions. 

We have defined a semi-exact log-Euler scheme and analysed the relevant convergence, when $\rho_2=0$. However, the same scheme does not apply directly, when $\rho_2\neq 0$, because if we write $\int_{0}^{T}\sqrt{V_t}\mathrm{d}W_{t}^{3}=\sqrt{\int_{0}^{T}V_t\mathrm{d}}\tilde{N}$, then the standard normal random variable $\tilde{N}$ depends on $(r_{t})_{t\geq 0}$ with an unknown correlation. 

Let $X_t=\ln S_t$, and following the similar steps, we have
\begin{align*}
X_{(i+1)h}&=X_{ih}+\left[\int_{ih}^{(i+1)h}r_{t}\mathrm{d}t+\left(\frac{\rho_1 k}{\sigma}-\frac{1}{2}\right)\int_{ih}^{(i+1)h}V_{t}\mathrm{d}t+\frac{\rho_1}{\sigma}\left(V_{(i+1)h}-V_{ih}-k\theta h\right)\right.\nonumber\\
&\quad\left.+\rho_2\int_{ih}^{(i+1)h}\sqrt{V_{t}}\mathrm{d}W_{t}^{3}+\sqrt{1-\rho_1^{2}-\rho_2^{2}}\int_{ih}^{(i+1)h}\sqrt{V_{t}}\mathrm{d}W_{t}^{2}\right]
\end{align*}
for $i=0,1,..,T/h$. The numerical scheme we propose is of the form
\begin{align*}
\tilde{X}_{(i+1)h}&=\tilde{X}_{ih}+\tilde{r}_{ih}h+\left(\frac{\rho_1 k}{\sigma}-\frac{1}{2}\right)V_{ih}h+\frac{\rho_1}{\sigma}\left(V_{(i+1)h}-V_{ih}-k\theta h\right)\\
&\quad+\rho_2\sqrt{V_{ih}}\Delta W^3+\sqrt{1-\rho_1^{2}-\rho_2^{2}}\sqrt{V_{ih}}\Delta W^{2}
\end{align*}
where $V$ is sampled exactly. 

\begin{thm} \label{thm5.1}
Suppose that Assumptions \ref{as1} and \ref{as2} are satisfied. Suppose that there exists $p>1$, such that $\mathbb{E}\left(e^{-2p\int_{0}^{T}r_{t}\mathrm{d}t}\right)< \infty$ and $\sup_h\mathbb{E}\left(e^{-2p\sum_{i=0}^{T/h-1}\tilde{r}_{ih}h}\right)< \infty$. Then, we have
\[
\mathbb{E}\left[\left(e^{-\int_{0}^{T}r_{t}\mathrm{d}t}P(S_{T})\right)^2\right]<\infty
\]
and
\[
\mathbb{E}\left[\left(e^{-\int_{0}^{T}r_{t}\mathrm{d}t}P(S_{T})-e^{-\sum_{i=0}^{T/h-1}\tilde{r}_{ih}h}P(\tilde{S}_{T}^h)\right)^2\right]=O(h).
\]
\end{thm}
\begin{proof}
The proof is very similar to the proof of Theorem \ref{thm1}. It suffices to prove that
\begin{equation}
\mathbb{E}\left[\left|\ln S_{T}-\ln \tilde{S}_{T}^{h}\right|^{q}\right]=O(h^{q/2}), \quad \forall q\geq 2. \label{eq5.1}
\end{equation}
We obtain
\begin{align}
&\mathbb{E}\left[\left|\ln S_{T}-\ln \tilde{S}_{T}^{h}\right|^{q}\right]\nonumber\\
&\leq c_{q}\mathbb{E}\left[\left|\int_{0}^{T}r_{t}\mathrm{d}t-\sum_{i=0}^{T/h-1} \tilde{r}_{ih}h\right|^{q}\right]+c_{q}\mathbb{E}\left[\left|\int_{0}^{T}V_{t}\mathrm{d}t-\sum_{i=0}^{T/h-1} V_{ih}h\right|^{q}\right]\nonumber\\
&\quad +c_{q}\mathbb{E}\left[\left|\int_{0}^{T}\left(\sqrt{V_t}-\sqrt{V_{\eta(t)}}\right)\mathrm{d}W_t^2\right|^{q}\right]+c_q\mathbb{E}\left[\left|\int_{0}^{T}\left(\sqrt{V_t}-\sqrt{V_{\eta(t)}}\right)\mathrm{d}W_t^3\right|^{q}\right].\label{eq5.2}
\end{align}
where $\eta(t):=\max\{lh:lh\leq t, l=0,1,2,...\}$. The first and second terms of (\ref{eq5.2}) is $O(h^q)$, due to Lemma \ref{lemma1}. The third and fourth terms can be estimated to be $O(h^{q/2})$. Specifically, Corollary 2.14 in Hutzenthaler, Jentzen and Noll \cite{HJN} implies that for any $q>0$, there exists a constant $L$ (for all $s,t\in [0,T]$) such that
\[
\mathbb{E}\left[\left|\sqrt{V_{t}}-\sqrt{V_{s}}\right|^q\right]\leq L|t-s|^{q/2}.
\]
Then, using Burkholder-Davies-Gundy inequality, we have
\begin{align*}
\mathbb{E}\left[\left|\int_{0}^{T}\left(\sqrt{V_t}-\sqrt{V_{\eta(t)}}\right)\mathrm{d}W_t^3\right|^{q}\right]&\leq c_q \mathbb{E}\left[\left|\int_{0}^{T}\left(\sqrt{V_t}-\sqrt{V_{\eta(t)}}\right)^2\mathrm{d}t\right|^{q/2}\right]\\
&\leq c_q\int_{0}^{T}\mathbb{E}\left(\left|\sqrt{V_t}-\sqrt{V_{\eta(t)}}\right|^{q}\right)\mathrm{d}t=O(h^{q/2}).
\end{align*}
Therefore, we have proved (\ref{eq5.1}), and the proof is complete.
\end{proof}
The convergence rate $O(h)$ in Theorem \ref{thm5.1}, corresponding to the rate $1/2$ in the $L^2$ norm, is optimal, because the same scheme applied to the standard Heston model yields the same convergence rate (Zheng \cite{CZ1}). As discussed, this rate may not be sufficiently high for the unbiased estimation. However, the development of a method of higher order is difficult, so we leave it for future research. Nevertheless, our result is still useful, because we can easily combine this scheme with the Multilevel Monte Carlo in Giles\cite{Gi}. 

\section{Applications}
In this section, we apply our results from Section \ref{sect2} to several well-known interest rate models in finance, including the CIR model, the Hull-White model and the Black-Karasinski model. The option payoff $P$ we consider satisfies either Assumption \ref{as2} or is the payoff of a digital option.

\subsection{Heston-CIR model}
The CIR model, which was introduced by Cox, Ingersoll and Ross \cite{CIR}, is represented as
\[
\mathrm{d}r_{t}=\alpha(\beta-r_{t})\mathrm{d}t+\gamma \sqrt{r_{t}}\mathrm{d}W_{t}^{3},
\]
where $\alpha,\beta,\gamma, r_0>0$. It is known that $r_{t}$ follows a scaled noncentral chi-squared distribution given $r_{u},u\in [0,t)$, i.e.,
\[
r_{t}\overset{\mathrm{d}}{=}\frac{\gamma^{2}(1-\mathrm{e}^{-\alpha(t-u)})}{4\alpha}\chi_{d}^{2}\left(\frac{4\alpha\mathrm{e}^{-k(t-u)}}{\sigma^{2}(1-\mathrm{e}^{-\alpha(t-u)})}r_{u}\right),
\]
where $\chi_{d}^{2}(\lambda)$ denotes a noncentral chi-squared random variable with degrees of freedom $d=\frac{4\alpha \beta}{\gamma^{2}}$  and noncentrality parameter $\lambda$ (see Glasserman \cite{G}).
For the exact simulation of $r_{t}$, i.e., $\hat{r}_{t}=r_{t}$, $t\in [0,T]$, we have $\sup_{t\in [0,T]}\mathbb{E}(r_{t}^q)<\infty$ for any $q>-\frac{2\alpha\beta}{\gamma^2}$. Hence, it is easy to verify that Assumption \ref{as1} is satisfied. As $\mathbb{P}(r_{t}\geq 0, \forall t\in [0,T])=1$, the moment condition in Theorem \ref{thm1} is satisfied and it follows from Theorem \ref{thm1} that 
\begin{equation}
\mathbb{E}\left[\left(e^{-\int_{0}^{T}r_{t}\mathrm{d}t}P(S_{T})-e^{-\sum_{i=0}^{T/h-1}\hat{r}_{ih}h}P(\hat{S}_{T}^h)\right)^2\right]=O(h^2).\label{eqtnh}
\end{equation}
Note that Theorem \ref{thm2} has the same assumptions as Theorem \ref{thm1}. Hence, for digital options, equation (\ref{eqtnh}) also holds.

Since the exact simulation of the CIR process can be time-consuming, there are several time-discrete schemes for the CIR process (see Alfonsi \cite{Alf} for discussions). Neuenkirch and Szpruch \cite{NS} showed that the BEM scheme and the drift-implicit Milstein scheme preserve the nonnegativity of the CIR process, i.e., $\mathbb{P}(\hat{r}_{t}\geq 0, \forall t\in [0,T])=1$ and both are strongly convergent with order one when $\frac{2\alpha\beta}{\gamma^2}>3$. Specifically, the BEM scheme can be written as
\[
\hat{x}_{ih}=\hat{x}_{(i-1)h}+\frac{\alpha}{2}\left(\left(\beta-\frac{\gamma^2}{4\alpha}\right)\hat{x}_{ih}^{-1}-\hat{x}_{ih}\right)h+\frac{\gamma}{2}(W_{ih}-W_{(i-1)h}),
\]
where $\hat{r}_{ih}=\hat{x}_{ih}^2$. Proposition 3.1 in Neuenkirch and Szpruch \cite{NS} demonstrated that
\[
\mathbb{E}\left[\max_{i=1,..,T/h}(\hat{r}_{ih}-r_{ih})^p\right]<c_{n}h^p,
\]
if $2\leq p<\frac{4}{3}\frac{\alpha\beta}{\gamma^2}$. The drift-implicit Milstein scheme is
\[
\hat{r}_{ih}=\hat{r}_{(i-1)h}+\alpha(\beta-\hat{r}_{ih})h+\gamma\sqrt{\hat{r}_{(i-1)h}}(W_{ih}-W_{(i-1)h})+\frac{\gamma^2}{4}((W_{ih}-W_{(i-1)h})^2-h).
\]
Lemma 4.1 in Neuenkirch and Szpruch \cite{NS} guaranteed that
\[
\max_{i=1,..,T/h}\mathbb{E}\left|\hat{r}_{ih}-r_{ih}\right|<ch,
\]
if $\frac{\alpha\beta}{\gamma^2}>\frac{3}{2}$. These results indicate that Assumption \ref{as1} might be satisfied; thus, (\ref{eqtnh}) might hold for both the BEM scheme and the drift-implicit Milstein scheme applied to the CIR process.

\subsection{Heston-Hull-White model}
The Hull-White model (Hull and White \cite{HW}) is of the form
\begin{equation}
\mathrm{d}r_{t}=\alpha(\beta(t)-r_{t})\mathrm{d}t+\gamma\mathrm{d}W_{t}^{3},\label{HWhite}
\end{equation}
where $\alpha,\gamma>0$, $r_0\in\mathbb{R}$ and $\beta:[0,T]\rightarrow \mathbb{R}^{+}$ is continuous. Given $r_u,u\in [0,t)$, the interest rate $r_t$ is normally distributed with mean 
\[
e^{-\alpha(t-u)}r_{u}+\alpha\int_{u}^{t}e^{-\alpha(t-s)}\beta(s)\mathrm{d}s
\]
and variance
\[
\frac{\gamma^2}{2\alpha}(1-e^{-2\alpha(t-u)})
\]
(see Glasserman \cite{G}, p109). In practice, it is often the case that $\beta(t)$ has a simple structure so that it is convenient to simulate $r_{t}$ exactly. Let $\hat{r}_{t}=r_{t}$, $t\in [0,T]$. Since it holds that $\sup_{t\in [0,T]}\mathbb{E}[|r_{t}|^q]<\infty$ for any $q\geq 0$, Assumption \ref{as1} is satisfied. Furthermore, we have $\mathbb{E}\left(e^{-2p\int_{0}^{T}r_{t}\mathrm{d}t}\right)< \infty$ for any $p\in \mathbb{R}$ (see Glasserman \cite{G}, p111). The lemma below shows the boundedness of $\sup_h\mathbb{E}\left(e^{-2p\sum_{i=0}^{T/h-1}r_{ih}h}\right)$.

\begin{lemma}
Let $(r_t)_{t\in [0,T]}$ satisfy the Hull-White model (\ref{HWhite}). Let $h=T/n$, $n\in \mathbb{N}^+$. For any $p\in \mathbb{R}$, we have
\[
\sup_h\mathbb{E}\left(e^{-2p\sum_{i=0}^{T/h-1}r_{ih}h}\right)< \infty.
\]
\end{lemma}
\begin{proof}
It is easy to show that there exists a constant $c$, which is independent of $i$ and $h$, such that 
\[
\max\left(\left|2p\alpha\int_{(i-1)h}^{ih}e^{-\alpha(ih-s)}\beta(s)\mathrm{d}s\right|, \left|\frac{p^2\gamma^2}{\alpha}(1-e^{-2\alpha h})\right|\right)<ch.
\]
Moreover, it is known that if $X$ is normally distributed with mean $\mu$ and variance $\sigma^2$, then $\mathbb{E}(e^{X})=e^{\mu+\frac{1}{2}\sigma^2}$. Let
\begin{align*}
h_1&:=h\leq h\\
h_2&:=h+e^{-\alpha h}h_1\leq 2h\\
...\\
h_i&:=h+e^{-\alpha h}h_{i-1}\leq ih\\
...\\
h_{n}&:=h+e^{-\alpha h}h_{n-1}\leq T
\end{align*}
where $n=T/h$ and $\alpha>0$. Then, we have
\begin{align*}
&\ln\mathbb{E}(e^{-2ph_i r_{(n-i)h}}|r_{(n-i-1)h})\\
&\leq -2ph_i e^{-\alpha h}\cdot r_{(n-i-1)h}+ch\cdot h_i+ch\cdot h_i^2\\
&\leq -2ph_i e^{-\alpha h}\cdot r_{(n-i-1)h}+cih^2+ci^2 h^3
\end{align*}
for any $i=1,2,...,n-1$. Therefore, based on conditional expectations, we can calculate recursively that
\begin{align*}
&\mathbb{E}\left(e^{-2p\sum_{i=0}^{T/h-1}r_{ih}h}\right)\\
&\leq\mathbb{E}\left(e^{-2p\sum_{i=0}^{T/h-2}r_{ih}h}\cdot e^{-2ph_1 e^{-\alpha h}r_{(T/h-2)h}}\right)e^{ch^2+ch^3}\\
&\leq \mathbb{E}\left(e^{-2p\sum_{i=0}^{T/h-3}r_{ih}h}\cdot e^{-2ph_2 e^{-\alpha h}r_{(T/h-3)h}}\right)e^{c\sum_{i=1}^{2}ih^2+c\sum_{i=1}^{2}i^2 h^3}\\
&\leq ...\\
&\leq \left(e^{-2pr_{0}h}\cdot e^{-2ph_{T/h-1} e^{-\alpha h}r_0}\right) e^{c\sum_{i=1}^{T/h-1}ih^2+c\sum_{i=1}^{T/h-1}i^2 h^3}\\
&= e^{-2pr_{0}h_{T/h}} e^{c\sum_{i=1}^{T/h-1}ih^2+c\sum_{i=1}^{T/h-1}i^2 h^3}\leq c_p,
\end{align*}
where $c_p$ is independent of $h$. This means that $\sup_h\mathbb{E}\left(e^{-2p\sum_{i=0}^{T/h-1}r_{ih}h}\right)$ is also bounded and the proof is complete.
\end{proof}
Then, all assumptions in Theorems \ref{thm1} or \ref{thm2} are satisfied, and we conclude that 
\[
\mathbb{E}\left[\left(e^{-\int_{0}^{T}r_{t}\mathrm{d}t}P(S_{T})-e^{-\sum_{i=0}^{T/h-1}\hat{r}_{ih}h}P(\hat{S}_{T}^h)\right)^2\right]=O(h^2).
\] 

\subsection{Heston-Black-Karasinski model}
The Black-Karasinski model (Black and Karasinski \cite{BK}) can be written as
\[
\mathrm{d}\ln r_{t}=(\beta(t)-\alpha\ln r_{t})\mathrm{d}t+\gamma\mathrm{d}W_{t}^{3},
\]
where $\alpha,\gamma, r_0>0$ and $\beta:[0,T]\rightarrow \mathbb{R}^{+}$ is continuous. It follows from It$\hat{o}$'s formula that
\[
\mathrm{d} r_{t}=r_{t}\left(\beta(t)+\frac{\gamma^2}{2}-\alpha\ln r_{t}\right)\mathrm{d}t+\gamma r_{t}\mathrm{d}W_{t}^{3}.
\]
Given $r_u,u\in [0,t)$, the random variable $r_t$ has a lognormal distribution (see Brigo and Mercurio \cite{BM}); hence, $r_{t}$ is usually simulated exactly. Specifically, given $r_u,u\in [0,t)$, the logarithm of the interest rate $\ln r_t$ is normally distributed with mean 
\[
e^{-\alpha(t-u)}\ln r_{u}+\int_{u}^{t}e^{-\alpha(t-s)}\beta(s)\mathrm{d}s
\]
and variance
\[
\frac{\gamma^2}{2\alpha}(1-e^{-2\alpha(t-u)}).
\]
Let $\hat{r}_{t}=r_{t}$, $t\in [0,T]$. Since each moment of a lognormal random variable is finite, we have $\sup_{t\in [0,T]}\mathbb{E}(r_{t}^{q})<\infty$ for any $q\in \mathbb{R}$. It holds from the continuity of $\beta$ and the Cauchy-Schwarz inequality that
\begin{align*}
&\sup_{t\in [0,T]}\mathbb{E}\left[\left|r_{t}\left(\beta(t)+\frac{\gamma^2}{2}-\alpha\ln r_{t}\right)\right|^{q}\right]\\
&\leq c_{q}\sup_{t\in [0,T]}\mathbb{E}(r_{t}^{q})+c_{q}\sup_{t\in [0,T]}\mathbb{E}\left[\left|r_{t}\ln r_{t}\right|^{q}\right]\\
&\leq c_{q}\sup_{t\in [0,T]}\mathbb{E}(r_{t}^{q})+c_{q}\sqrt{\sup_{t\in [0,T]}\mathbb{E}(r_{t}^{2q})\cdot\sup_{t\in [0,T]}\mathbb{E}\left[\left|\ln r_{t}\right|^{2q}\right]}<\infty,
\end{align*}
for any $q\geq 1$. Thus, Assumption \ref{as1} is satisfied. As $r_{t}$ is nonnegative, the moment condition is automatically satisfied, and we obtain from Theorems \ref{thm1} or \ref{thm2} that
\[
\mathbb{E}\left[\left(e^{-\int_{0}^{T}r_{t}\mathrm{d}t}P(S_{T})-e^{-\sum_{i=0}^{T/h-1}\hat{r}_{ih}h}P(\hat{S}_{T}^h)\right)^2\right]=O(h^2).
\]

\section{Numerical results}
In this section, we conduct numerical experiments to verify the convergence rate derived in Section \ref{sect2} and then evaluate the efficiency of the log-Euler scheme we develop combined with Rhee and Glynn's unbiased estimators. 

We consider the Heston model with three different interest rate models: the CIR model, the Hull-White model and the Black-Karasinski model. For the CIR model, we focus on two methods to simulate the path: the exact simulation method and the BEM scheme from Neuenkirch and Szpruch \cite{NS}. For the Hull-White model and the Black-Karasinski model, we simulate the paths exactly. We evaluate the convergence rate of the error
\[
Err(h):=\mathbb{E}\left[\left(e^{-\int_{0}^{T}r_{t}\mathrm{d}t}P(S_{T})-e^{-\sum_{i=0}^{T/h-1}\hat{r}_{ih}h}P(\hat{S}_{T}^h)\right)^2\right],
\]
where $P(S_{T})$ is the payoff of an option and $h$ is the step size. Three types of options are considered:
\begin{itemize}
\item European put option $P(S_{T}):=\max(K-S_{T},0)$; 
\item European call option $P(S_{T}):=\max(S_{T}-K,0)$;
\item Digital call option $P(S_{T}):=1_{S_{T}>K}$.
\end{itemize}
The model parameters are available in Table \ref{table1} and we set $T=1$ and $S_{0}=K=1$ for all cases. These options are interesting for the reason below: a European put option has a bounded and continuous payoff, a European call options has an unbounded payoff and a digital option has a discontinuous payoff. All experiments are performed in Matlab. 

Note that the payoff of a European put option satisfies Assumption \ref{as2}, and according to Theorem \ref{thm1}, the theoretical convergence rate should be $2$. For the digital call option, Theorem \ref{thm2} guarantees that the rate is also $2$ if the approach of conditional expectations is used. The European call option is out of the scope of our analysis. 

\begin{table} 
\begin{center}
\begin{tabular}{ c | c | c | c | c | c | c | c | c | c | c }
\hline 
& $k$ & $\theta$ & $\sigma$ & $\rho$ & $\alpha$ & $\beta$ & $\gamma$ & $r_{0}$ & $V_{0}$ & $S_{0}$\tabularnewline
\hline 
CIR-exact & $2.8$ & $0.05$ & $0.25$ & $0.5$ & $1.2$ & $0.06$ & $0.25$ & $0.05$ & $0.04$ & $1$\tabularnewline
CIR-BEM  & $3$ & $0.04$ & $0.25$ & $0.5$ & $3.5$ & $0.06$ & $0.25$ & $0.05$ & $0.04$ & $1$\tabularnewline
HW &  $2.8$ & $0.05$ & $0.25$ & $0.5$ & $1.2$ & $0.06$ & $0.5$ & $0.05$ & $0.04$ & $1$\tabularnewline
BK &  $2.8$ & $0.05$ & $0.25$ & $0.5$ & $1.2$ & $0.06$ & $0.25$ & $0.05$ & $0.04$ & $1$\tabularnewline
\hline 
\end{tabular}
\par\end{center}
\caption{Parameters of the Heston model with stochastic interest rates, where CIR-exact, CIR-BEM, HW and BK represent the CIR model with exact simulation, the CIR model simulated through the BEM method, the Hull-White model and the Black-Karasinski model, respectively.} \label{table1}
\end{table}

Figure \ref{Fig} plots $\log_{2}(Err(h))$ against $-\log_{2}(h)$, with $h=2^{-n}$, $n=0,1,2,...,7$. Here, the exact values of  $S_{T}$ and $\int_{0}^{T}r_{t}\mathrm{d}t$ are approximated by using the Euler scheme in Section \ref{Hestonstr} based on a very small step size $2^{-10}$. For the BEM method, the exact path of $(r_{t})_{t\in [0,T]}$ needs an additional approximation also with step size $2^{-10}$, which shares the same Brownian motion path with the corresponding $(\hat{r}_{ih})_{i=1,2,...,T/h}$. To estimate $Err(h)$, the number of Monte Carlo samples for each quantity in each model is at least $0.5$ million, so that the standard deviation of the estimator of $Err(h)$ is typically less than $1\%$ of the estimated value of $Err(h)$. As illustrated in Figure \ref{Fig}, the convergence rate in all cases is $2$, for all models and all option payoffs that we consider, which is consistent with the theoretical convergence rate. Furthermore, although the European call option does not fall into the scope of our analysis, it still has a desired convergence rate. This implies that our convergence result might be extended for more types of options.

\begin{figure}
\centering
\includegraphics[scale=0.5]{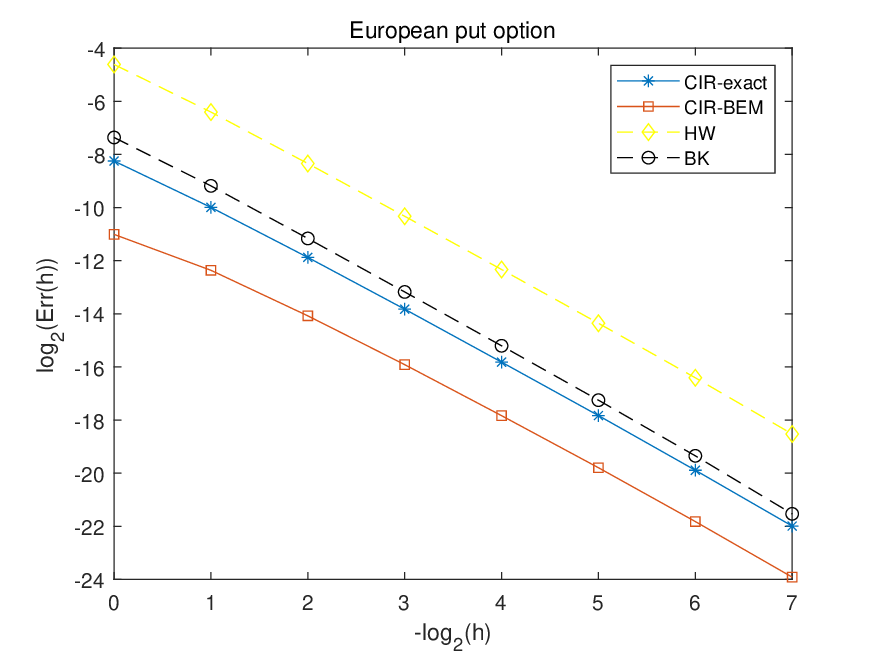}
\includegraphics[scale=0.5]{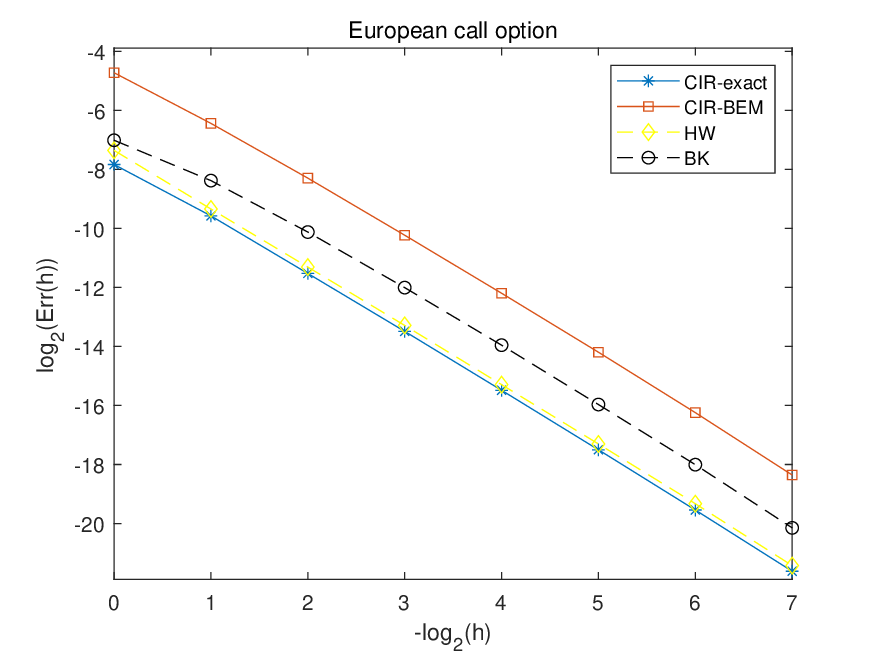}
\includegraphics[scale=0.5]{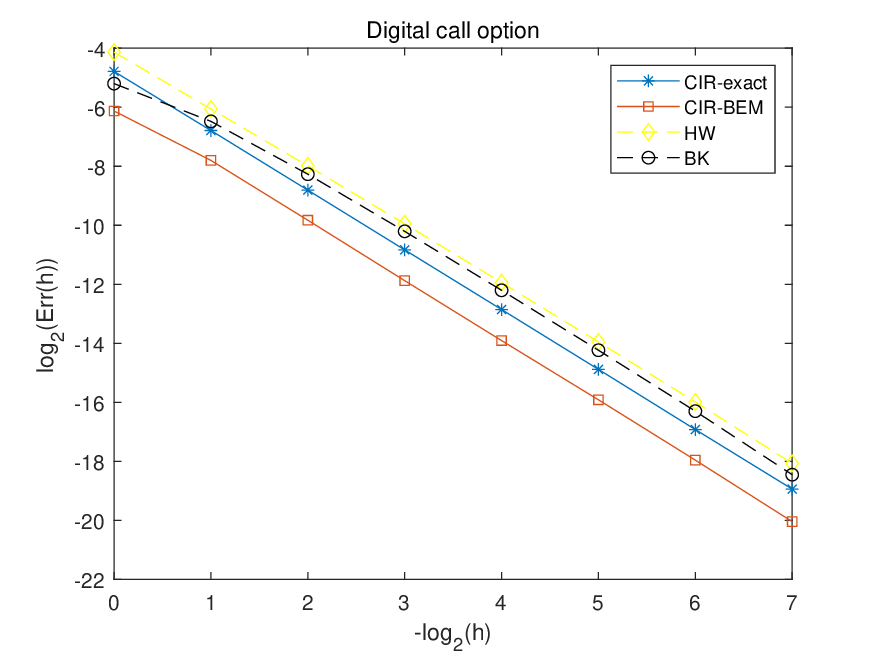}
\caption{Convergence rate for the Heston model with stochastic interest rates. The model parameters are from Table \ref{table1}. }
\label{Fig}
\end{figure}

Next, we incorporate the semi-exact log-Euler scheme into Rhee and Glynn's unbiased estimators. As discussed in Section \ref{Sect3}, the implementation of  unbiased estimation requires setting a distribution of $\mathcal{N}$. For the coupled sum estimator $Z$, we simply take $\mathbb{P}(\mathcal{N}\geq n)=2^{-3n/2}$, $n\in \mathbb{N}$, such that (\ref{var}) and (\ref{ct}) are finite. Hence, $Z$ is unbiased with a finite variance and finite computational time. Table \ref{table2} reports the root mean square error (RMSE) and the computational time (in seconds, for European put options) of $Z$ based on $1$ million samples. For some applications, either the variance or the computational time of $Z$ can be infinite; see Zheng, Blanchet and Glynn \cite{ZBG}. We see from Table \ref{table2} that all these quantities are finite, which again coincides with the theory. In particular, the computational time for all three options is finite. Therefore, the log-Euler scheme we develop is well-suited to the framework of Rhee and Glynn's unbiased estimators. Moreover, compared with the result from Figure \ref{Fig}, we observe that a method for an interest rate model with a large RMSE of $Z$ tends to have a large $Err(h)$. For example, for the European call option, CIR-BEM has the largest RMSE and $Err(h)$, then follows by BK and HW, and CIR-exact has the smallest errors. On the other hand, the digital option generally has a larger RMSE compared with the European call or put options. This may be because that the payoff function of the digital option is discontinuous at $S_T=K$, and $|P(\hat{S}^h_T)-P(S_T)|=1$ can occur when $S_T\approx K$. 

For the single term estimator $\tilde{Z}$, the result is illustrated in Table \ref{table3}, where we set the distribution of $\mathcal{N}$ to be 
\[
p_n=(1-2^{-3/2})2^{-3n/2}.
\]
It is optimal that $p_n=O(2^{-3n/2})$ (Zheng, Pan and Wang \cite{ZPW}) and $p_n$ should satisfy $\sum_{n=0}^{\infty}p_n=1$. As we see from Table \ref{table3}, the RMSE and computational time for all interest rate models and all types of options are finite. Furthermore, the digital option generally has the largest RMSE, followed by the European call option and the European put option, which is analogous to the result from Table \ref{table2}. In addition, the efficiency of the single term estimator seems similar to that of the coupled sum estimator. However, the former has a simpler structure, and it may be more favorable in practical applications. 

\begin{table} 
\begin{center}
\begin{tabular}{ c | c | c | c | c}
\hline 
& RMSE-put & RMSE-call & RMSE-digital& Comput time  \tabularnewline
\hline 
CIR-exact & $1.21\times 10^{-4}$ & $3.39\times 10^{-4}$ & $4.22\times 10^{-4}$ & $8.78$ \tabularnewline
CIR-BEM  & $1.96\times 10^{-4}$  & $0.0014$ & $4.45\times 10^{-4}$& $8.37$ \tabularnewline
HW &  $4.16\times 10^{-4}$  & $4.57\times 10^{-4}$ & $6.32\times 10^{-4}$ & $5.72$ \tabularnewline
BK &  $1.78\times 10^{-4}$  & $5.31\times 10^{-4}$ & $8.32\times 10^{-4}$ & $6.17$ \tabularnewline
\hline 
\end{tabular}
\par\end{center}
\caption{The RMSE and computational time (in seconds) of the coupled sum estimator $Z$ based on $10^6$ samples. The model parameters are from Table \ref{table1}.} \label{table2}
\end{table}

\begin{table} 
\begin{center}
\begin{tabular}{ c | c | c | c | c}
\hline 
& RMSE-put & RMSE-call & RMSE-digital& Comput time  \tabularnewline
\hline 
CIR-exact & $1.63\times 10^{-4}$ & $3.57\times 10^{-4}$ & $4.06\times 10^{-4}$ & $6.74$ \tabularnewline
CIR-BEM  & $1.48\times 10^{-4}$  & $3.50\times 10^{-4}$ & $3.38\times 10^{-4}$ & $7.84$ \tabularnewline
HW &  $4.15\times 10^{-4}$  & $4.76\times 10^{-4}$ & $7.26\times 10^{-4}$ & $5.61$ \tabularnewline
BK &  $2.13\times 10^{-4}$  & $4.32\times 10^{-4}$ & $4.57\times 10^{-4}$ & $5.99$ \tabularnewline
\hline 
\end{tabular}
\par\end{center}
\caption{The RMSE and computational time (in seconds) of the single term estimator $\tilde{Z}$ based on $10^6$ samples. The model parameters are from Table \ref{table1}.} \label{table3}
\end{table}

\section{Conclusion}

In this article, we develop a semi-exact log-Euler scheme for the Heston model with stochastic interest rates and we analyse the relevant convergence rate in the $L^2$ norm. The SDEs of the Heston model with stochastic interest rates can be divided into two components: the Heston component and the interest rate component. Under mild assumptions on the interest rate component, we show that the convergence rate in the $L^2$ norm is one, which enables us to easily incorporate the log-Euler scheme into Rhee and Glynn's unbiased estimators. In our analysis, we consider two types of payoffs of options: the bounded and Lipschitz payoff and the payoff of a digital option. Furthermore, we demonstrate that the log-Euler scheme and the convergence analysis apply to a large class of interest rate models, including the well-known CIR model, Hull-White model and Black-Karasinski model.

There are two directions of extensions that might be interesting: the log-Euler scheme we consider is based on the assumption that the driven Brownian motion $W^3$ for the interest rate model is independent of the driven Brownian motions $W^1$ and $W^2$ for the SDEs of $S$ and $V$. One direction is to extend the scheme to the case without this assumption, i.e., the full constant correlation case, and analyse the convergence rate. The other direction is to extend the payoff $P$ to more complicated cases, such as those in Cozma, Mariapragassam and Reisinger \cite{CMR}.

\section*{Acknowledgements}
This research is supported by National Natural Science Foundation of China (No.11801504).

\end{document}